\documentclass[11pt]{article}

\usepackage[T1]{fontenc}
\usepackage[utf8]{inputenc}
\usepackage{lmodern}
\usepackage{microtype}
\usepackage{geometry}
\usepackage{amsmath,amssymb,amsthm,mathtools}

\usepackage{enumitem}
\usepackage{booktabs}

\usepackage{algorithm}
\usepackage{algpseudocode}

\usepackage[hidelinks]{hyperref}
\usepackage[nameinlink,noabbrev]{cleveref}

\theoremstyle{definition}
\newtheorem{definition}{Definition}[section]

\newtheorem{remark}{Remark}[section]
\theoremstyle{plain}
\newtheorem{proposition}[definition]{Proposition}
\newtheorem{theorem}[definition]{Theorem}
\newtheorem{conjecture}[definition]{Conjecture}

\newcommand{\Sig}{\Sigma}

\newcommand{\Thy}{T}
\newcommand{\Struct}{\mathcal{A}}
\newcommand{\Univ}{\left|\Struct\right|}
\newcommand{\N}{\mathbb{N}}
\newcommand{\Pow}{\mathcal{P}}

\newcommand{\forces}{\Vdash}              
\newcommand{\kforces}{\Vdash_k}           
\newcommand{\kbforces}{\Vdash_{k,b}}      

\newcommand{\cost}{\mathsf{cost}}

\newcommand{\Cert}{\mathsf{Cert}}
\newcommand{\Obl}{\mathsf{Obl}}

\title{Ultraconstructive Model Theory via Bounded Adversarial Finite Structures}
\author{
  Mirco A. Mannucci\\
  \textit{HoloMathics LLC}\\
  \texttt{mirco@holomathics.com}
}
\date{\today}

\begin{document}
\maketitle

\begin{abstract}
Ultraconstructive Model Theory (UCMT) replaces idealized satisfaction, at finite computational scale,
by bounded adversarial survival.  A finite partial structure is tested by an Opponent (Devil) drawing
legal challenges from a bounded attack surface, repaired by a Builder (God) through legal replies, and
certified by a symbolic Judge.  This defines a bounded forcing relation $\kbforces$ and three outcomes:
$\textsc{God\_Wins}$, $\textsc{Draw}$, and $\textsc{Devil\_Wins}$.  We prove a self-contained finite
metatheory for bounded episodes: termination, Judge-relative soundness of $\textsc{God\_Wins}$, soundness
of $\textsc{Devil\_Wins}$ as a bounded obstruction certificate, exhaustive bounded completeness over finite
completion spaces, and a neural-admissibility theorem showing that learned policies preserve logical
soundness when they only select among legal symbolic moves.  A prototype implementation, ADAMANTIUM,
realizes the God--Devil--Judge loop on controlled cyclic tasks: a $3$-element instance yields
$\textsc{God\_Wins}$, a $2$-element impossibility yields a certified $\textsc{Devil\_Wins}$ obstruction
(128 completions checked, zero winning completions, no budget exhaustion, Judge verified), and a minimal
judged co-training loop updates both neural policies over legal moves.  The connection to MTU-II /
Esenin--Volpin semantics is stated only as a conditional bridge.  The experiments are deliberately tiny;
we do not claim general finite model-finding performance, first-order completeness, or a complete theorem
prover.
\end{abstract}


\noindent\textbf{ACM CCS Concepts.}
Theory of computation~$\rightarrow$ Finite Model Theory; Theory of computation~$\rightarrow$ Logic and verification;
Computing methodologies~$\rightarrow$ Machine learning.

\noindent\textbf{arXiv subject classes.}
cs.LO; math.LO; cs.AI.


\tableofcontents

\section{Introduction}

Classical model theory evaluates a completed structure $\Struct$ by the ideal relation
$\Struct\models\Thy$.  UCMT asks for a finite, operational substitute: what can be certified when
structures are partial, resources are bounded, and the relevant tests are chosen adversarially?  The
central proposal is that, at finite computational scale, modelhood should be replaced by bounded
adversarial survival.  A partial structure is not declared a model once and for all; it survives a
specified attack surface, depth, and budget.

The formal object is a game between a Builder, an Opponent, and a symbolic Judge.  The Builder (God)
proposes legal semantic edits; the Opponent (Devil) selects legal challenges; the Judge is the only
source of logical verdicts.  The resulting bounded forcing notation $\Struct\kbforces\Thy$ means that
$\Struct$ survives the declared depth-$k$, budget-$b$ episode, not that it satisfies $\Thy$ in the
unbounded classical sense.  The three possible outcomes are $\textsc{God\_Wins}$, $\textsc{Draw}$, and
$\textsc{Devil\_Wins}$, the last requiring a bounded obstruction certificate rather than a mere timeout.

Two earlier strands motivate the definitions but are not needed for the finite results.  MTU-II
\cite{Mannucci2023MTU2} provides the bounded Kripke/Esenin--Volpin background: finite stages, bounded
term generation, and $k$-validity.  LOGAN \cite{Mannucci2025LOGAN} provides the adversarial interface:
depth-bounded logical probes, witnesses, and EF-style Opponents.  In the present paper these become a
self-contained finite game first; the bridge back to MTU-II is postponed to \Cref{sec:kcons-bridge} and
kept explicitly conditional.

\paragraph{Contributions.}
\begin{enumerate}[leftmargin=1.5em]
\item A finite bounded God--Devil--Judge semantics for partial finite structures, with attack surfaces,
obligations, certificates, and the bounded forcing relation $\kbforces$.
\item A self-contained finite metatheory of bounded episodes (\Cref{sec:metatheory}): termination,
Judge-relative soundness of $\textsc{God\_Wins}$, bounded-obstruction soundness of $\textsc{Devil\_Wins}$,
exhaustive bounded completeness over finite completion spaces, and neural admissibility.
\item A prototype apparatus, ADAMANTIUM, in which neural policies select only among legal symbolic moves
while the Judge certifies all outcomes.
\item Controlled cyclic experiments: a $3$-element $\textsc{God\_Wins}$, a certified
$2$-element $\textsc{Devil\_Wins}$ obstruction, and a minimal judged co-training loop updating both
neural policies.
\end{enumerate}

\paragraph{Claim boundary.}
The finite UCMT propositions are the mathematical content proved here.  We do not prove a general theorem
prover, a complete finite model finder, a full logic-driven GAN, or an equivalence with MTU-II.  The
experiments are intentionally small and are used to validate the semantics and certificate discipline, not
to establish scaling.

\section{Background I: Esenin--Volpin models as resource-bounded Kripke semantics}

This section restates the pieces of the MTU-II pipeline that UCMT will reuse: finite rooted frames,
bounded term generation, modified forcing, and bounded validity.

\subsection{Kripke frames, bounded depth, and the ``world as stage'' reading}

\begin{definition}[Finite rooted Kripke frame]
A \emph{finite rooted Kripke frame} is a triple $(W,\le,r)$ where:
\begin{itemize}[leftmargin=1.5em]
\item $W$ is a finite set of worlds,
\item $\le$ is a partial order (accessibility),
\item $r\in W$ is the root and $r\le w$ for all $w$ in the connected component.
\end{itemize}
The \emph{depth} of $w$ is the length of the longest chain from $r$ to $w$.
\end{definition}

\begin{remark}[Bounded depth is not cosmetic]
In ultrafinitism, allowing arbitrary depth silently reintroduces ``infinite verification''.
MTU-II enforces a bounded frame depth and then introduces $k$-validity precisely to avoid a collapse where
full (unbounded) Kripke validity becomes too restrictive under the altered forcing dynamics.
\end{remark}

\subsection{Volpin-style bounded term generation}

We need a formal handle for ``worlds are not globally term-closed.''

\begin{definition}[Term-generation operator]
Fix a signature $\Sig$ and a set of ground symbols available at a world $w$ (constants, and already-built elements).
A \emph{term-generation operator} is a map
\[
\mathsf{Gen}_w : \Pow(D(w)) \to \Pow(D(w))
\]
intended to represent ``one generation step'' of term construction. One can iterate:
$\mathsf{Gen}_w^{(0)}(S)=S$ and $\mathsf{Gen}_w^{(n+1)}(S)=\mathsf{Gen}_w(\mathsf{Gen}_w^{(n)}(S))$.
\end{definition}

\begin{definition}[Feasible terms at depth $k$ (schema)]
Given a seed set $S_w\subseteq D(w)$, define the \emph{feasible closure up to depth $k$}:
\[
\mathrm{Cl}_w^{\le k}(S_w) \;=\; \bigcup_{i=0}^k \mathsf{Gen}_w^{(i)}(S_w).
\]
A key MTU-II move is that forcing at $w$ only quantifies over objects in $\mathrm{Cl}_w^{\le k}$ (or an analogous
resource-bounded domain), rather than over an unbounded term model.
\end{definition}

\subsection{Esenin--Volpin forcing and $k$-validity (spine)}

\begin{definition}[Esenin--Volpin model (condensed schema)]
An \emph{Esenin--Volpin model} over $\Sig$ consists of:
\begin{itemize}[leftmargin=1.5em]
\item a finite rooted frame $(W,\le,r)$,
\item a domain assignment $w\mapsto D(w)$ with monotonicity: $w\le w'\Rightarrow D(w)\subseteq D(w')$,
\item interpretations of $\Sig$ at each world compatible with extension,
\item a forcing relation $w\forces \varphi$ whose clauses are modified to respect feasible closure / bounded generation.
\end{itemize}
(See \cite{Mannucci2023MTU2} for the precise forcing clauses and motivation.)
\end{definition}

\begin{definition}[$k$-validity]
For a formula $\varphi$, define $\kforces \varphi$ to mean ``$\varphi$ holds at all worlds up to depth $k$''
(or the corresponding MTU-II bounded validity notion).
\end{definition}

\begin{definition}[$k$-consistency (proof-theoretic companion)]
A theory $\Thy$ is \emph{$k$-consistent} if no contradiction has a proof of depth $\le k$ in the chosen system.
MTU-II relates bounded forcing/validity to $k$-consistency as the proof-theoretic control parameter.
\end{definition}

\begin{remark}[Why UCMT needs this spine]
UCMT will preserve the \emph{world-as-stage} semantics and boundedness knobs, but it will replace the implicit
``test by extension'' behavior with an explicit Opponent that selects which bounded instances are checked and
returns concrete witnesses.
\end{remark}

\section{Background II: LOGAN and explicit EF-Opponents}

LOGAN \cite{Mannucci2025LOGAN} provides the operational interface we need:
depth-bounded logical probes, witness extraction, and a builder/repair loop.

\subsection{Depth as a logical budget: quantifier rank and EF rounds}

\begin{definition}[Depth parameter]
In LOGAN, ``depth'' can be realized as:
\begin{itemize}[leftmargin=1.5em]
\item quantifier rank $\le k$ (FO),
\item $k$-round Ehrenfeucht--Fra\"{\i}ss\'e games (structure indistinguishability),
\item bounded tester adaptivity depth.
\end{itemize}
For UCMT, depth will index the strength of the Opponent's challenge language.
\end{definition}

\begin{definition}[Witness contract (LOGAN-style)]
A \emph{witness} is a small, checkable object that certifies failure:
a tuple falsifying an axiom instance, a substructure pattern (e.g.\ an odd cycle),
or an EF transcript that pinpoints a mismatch. The key requirement is that witnesses are cheap to verify.
\end{definition}

\begin{remark}
This witness contract is what makes UCMT experimental: we can log witness types, repair costs, and resilience curves,
instead of only saying ``model / not model''.
\end{remark}

\section{UCMT: explicit adversaries over Volpin-style stages}

\subsection{Partial structures, unknowns, and revision policies}

\begin{definition}[Partial $\Sig$-structure with unknowns]\label{def:partial}
A partial $\Sig$-structure $\Struct$ has finite domain $\Univ$ and partial interpretations:
\begin{itemize}[leftmargin=1.5em]
\item for each relation symbol $R$, an interpretation $R^\Struct(\bar a)\in\{0,1,?\}$,
\item for each function symbol $f$, a partial map $f^\Struct:\Univ^n\rightharpoonup \Univ$ (possibly undefined).
\end{itemize}
\end{definition}

\begin{definition}[Revision policy and cost]
A \emph{revision policy} specifies which previously committed facts may be changed and at what cost.
We model this abstractly by a cost functional $\cost(\text{edit})\in\N$ and a bound (or penalty weight) on total cost.
Monotone policies forbid revision; revising policies permit it but log and penalize it.
\end{definition}

\subsection{Obligations as operational existentials}

\begin{definition}[Obligation store]
An \emph{obligation store} $\Obl$ is a finite multiset of demands of the form
\[
(\exists y\,\psi(\bar a,y),\bar a)
\]
where $\bar a\in\Univ$ and $\psi$ is an allowed formula schema. Discharging an obligation means producing $b$
and committing $\psi(\bar a,b)$ (possibly by extending $\Univ$).
\end{definition}

\subsection{Certificates: what has been validated must remain stable}

\begin{definition}[Certificate store]
A \emph{certificate store} $\Cert$ is a finite set of validated items, e.g.:
\begin{itemize}[leftmargin=1.5em]
\item validated axiom instances $\varphi(\bar a)$,
\item derived invariants (normal forms, closure properties),
\item ``do not break'' constraints produced during repair.
\end{itemize}
A repair step is \emph{admissible} if it preserves all items in $\Cert$ (or preserves them up to an allowed weakening).
\end{definition}

\begin{remark}
Without certificates, a repair loop can thrash: fix one witness by breaking something previously stable.
Certificates are the ultraconstructive analogue of ``proof obligations already discharged.''
\end{remark}

\subsection{Challenge languages, depth, and budget}

\begin{definition}[Challenge language $\mathcal{C}_k$]
A challenge language $\mathcal{C}_k$ is a set of Opponent moves bounded by depth $\le k$.
Canonical instances:
\begin{itemize}[leftmargin=1.5em]
\item \textbf{Equational/Horn mode:} instantiate universal Horn clauses; witness is a falsifying tuple.
\item \textbf{EF mode:} a $k$-round EF probe; witness is a transcript or distinguishing pattern.
\item \textbf{Tester mode:} bounded-query tests; witness is a local counterexample.
\end{itemize}
\end{definition}

\begin{definition}[Budget]
The budget $b$ bounds the number of challenges the Opponent may issue in one evaluation episode.
We treat $b$ as the experimental knob controlling the intensity of scrutiny at fixed depth $k$.
\end{definition}

\subsection{The UCMT game and bounded forcing}

\begin{definition}[UCMT game $G(\Sig,\Thy;k,b)$]
A position is $(\Struct,\Cert,\Obl,t)$ where $\Struct$ is partial, $\Cert$ certificates, $\Obl$ obligations, and $t\le b$ used budget.
Each round:
\begin{enumerate}[leftmargin=1.5em]
\item Opponent picks $c\in\mathcal{C}_k$ and returns either:
  \begin{itemize}[leftmargin=1.5em]
  \item a witnessed violation $w$ (an object proving that some induced constraint of $\Thy$ fails on current commitments), or
  \item a demand producing a new obligation.
  \end{itemize}
\item Builder applies admissible edits (respecting $\Cert$) to repair $w$ and/or discharge obligations,
      optionally extending $\Univ$, and may add newly validated items to $\Cert$.
\end{enumerate}
Builder wins if it survives all Opponent moves up to budget $b$ without violating certificate admissibility
(or exceeding allowed revision cost).
\end{definition}

\begin{definition}[UCMT bounded forcing / $(k,b)$-models]
We write $\Struct \kbforces \Thy$ if Builder has a winning strategy in $G(\Sig,\Thy;k,b)$ from the initial state.
A total finite $\Struct$ is a \emph{$(k,b)$-model} of $\Thy$ if $\Struct\kbforces \Thy$.
\end{definition}

\begin{definition}[Outcomes of a bounded episode]\label{def:outcomes}
We name three game outcomes (\Cref{sec:devilwins} develops the second):
\begin{itemize}[leftmargin=1.5em]
\item $\textsc{God\_Wins}$: Builder (God) survives the episode and produces a $(k,b)$-model candidate (with, in refute mode, a witness against the target claim).
\item $\textsc{Devil\_Wins}$: Opponent (Devil) produces a \emph{bounded obstruction certificate} (\Cref{def:obstruction}) --- not mere search failure.
\item $\textsc{Draw}$: the budget is exhausted with neither a certified construction nor a certified obstruction.
\end{itemize}
The conceptual roles are: Builder $=$ God, Opponent $=$ Devil, and the symbolic forcing/verification layer $=$ Judge.
\end{definition}

\subsection{Relation to MTU-II: UCMT as ``explicit adversary'' refinement}

\begin{proposition}[MTU-II as monotone special case (informal)]
If Builder is monotone (no revision), certificates are exactly the monotone forcing facts, and the Opponent is restricted
to ``advance to a successor world'' (rather than selecting arbitrary probes), then UCMT reduces to a bounded Kripke-style
evaluation regime and aligns with the MTU-II forcing/$k$-validity spine.
\end{proposition}

\begin{remark}
This is the correct conceptual hierarchy: MTU-II gives the frame/forcing boundedness; UCMT adds an explicit
LOGAN-style selection of which bounded instances are tested, and a witness/repair operational layer.
\end{remark}

\subsection{Parameter effects}

\begin{proposition}[Budget monotonicity]
If $\Struct \kbforces \Thy$ then $\Struct \Vdash_{k,b'} \Thy$ for all $b'\le b$.
\end{proposition}

\begin{conjecture}[Depth monotonicity in stable regimes]
Under universal Horn $\Thy$ and monotone builder policies, $\Struct \kbforces \Thy$ implies $\Struct \Vdash_{k',b}\Thy$
for $k'\le k$.
\end{conjecture}

\section{UCMT-CEGIS: witness-guided model synthesis}

\begin{algorithm}[h]
\caption{UCMT-CEGIS (witness-guided synthesis)}
\begin{algorithmic}[1]
\Require $\Sig,\Thy,k,b$, max domain size $n_{\max}$, builder policy $\pi$, opponent policy $\delta$
\State Initialize partial $\Struct \gets \Struct_0$, $\Cert\gets\emptyset$, $\Obl\gets\emptyset$
\For{$iter=1$ to \textbf{MaxIters}}
  \State Opponent($\delta$): issue up to $b$ depth-$k$ challenges; collect witnesses $W$ and demands $D$
  \State Builder($\pi$): for each $w\in W$, propose admissible repair edits minimizing $\cost$
  \State Builder($\pi$): add obligations from $D$ and attempt discharge (witness construction)
  \State Update $\Cert$ with newly validated instances (optional but recommended)
  \If{$W=\emptyset$ and $\Obl=\emptyset$}
     \State \Return $\Struct$ as a $(k,b)$-model candidate
  \EndIf
  \If{$\Univ>n_{\max}$} \State \Return \textbf{FAIL} \EndIf
\EndFor
\State \Return \textbf{UNKNOWN}
\end{algorithmic}
\end{algorithm}

\begin{remark}[Witness-shaped repairs]
Repairs should be ``local to the witness'': the edit set is constrained to the small region of $\Struct$
that the witness touches. This is exactly what makes repair cost measurable and supports ablation studies.
\end{remark}

\section{$\textsc{Devil\_Wins}$ and bounded obstruction certificates}
\label{sec:devilwins}

We define the Devil's win inside the UCMT game semantics, not merely as a
software flag.

\begin{definition}[Bounded obstruction certificate / $\textsc{Devil\_Wins}$]\label{def:obstruction}
Fix $(\Sig,\Thy,\text{goal})$, a finite domain $\Univ$, a depth $k$, and a budget $b$.
A \emph{bounded obstruction certificate} is a verifiable object showing that
\emph{every} legal Builder continuation within the stated finite domain and bounds
fails to reach a $(k,b)$-model achieving the goal --- or that the game has reached
a verified obstruction state (e.g.\ a committed configuration the Judge certifies
cannot be completed). $\textsc{Devil\_Wins}$ is declared exactly when such a
certificate is produced.
\end{definition}

\begin{remark}[Scope of a prototype obstruction certificate]\label{rem:obstruction-scope}
The certificate may be produced by \emph{exhaustive bounded checking}
on tiny domains (as in Demo~B, \Cref{sec:experiments}). This is \emph{not} a general impossibility theorem:
it is a bounded obstruction certificate for the \emph{stated} finite domain,
language fragment, depth $k$, and budget $b$. Lifting bounded obstruction to a
$k$-inconsistency claim requires the conditional bridge of
\Cref{sec:kcons-bridge} (in particular \Cref{prop:uniform-failure-refutation}),
including a complete bounded challenge language and a sound Judge.
\end{remark}

\section{Finite bounded metatheory of the episode}
\label{sec:metatheory}

The results of this section are \emph{self-contained}: they concern the finite,
fully enumerable version of the UCMT game that the prototype (\Cref{sec:apparatus})
actually runs, and they do not depend on the MTU-II bridge of
\Cref{sec:kcons-bridge}. All soundness statements are explicitly \emph{relative to
the symbolic Judge}: the Judge fixes which fragment (depth-$\le k$ ground instances
of $\Thy$ and the goal) is checked, and the metatheory is exactly as strong as that
specification.

\subsection{The bounded episode as a finite object}

\begin{definition}[Cells and completion space]\label{def:completion}
Fix a finite signature $\Sig$ and a finite domain $[n]=\{0,\dots,n-1\}$. The \emph{cells}
are the finitely many relation tuples $R(\bar a)$, function points $f(\bar a)$, and
constants $c$ ($\bar a\in[n]^{\mathrm{ar}}$). A partial $\Sig$-structure (\Cref{def:partial})
assigns each cell a value in its finite value set --- $\{0,1\}$ for relation cells, $[n]$ for
function/constant cells --- or leaves it \emph{unknown}. A \emph{completion} of a partial
$\Struct$ fills every unknown cell. The set $\mathrm{Comp}(\Struct)$ of total completions is
finite, with $|\mathrm{Comp}(\Struct)| = \prod_{c\ \text{unknown}} |\mathrm{val}(c)|$.
\end{definition}

\begin{definition}[Legal moves, legal replies, Judge]\label{def:legal}
At a state with partial structure $\Struct$:
\begin{itemize}[leftmargin=1.5em]
\item the \emph{legal Devil challenges} $L_D(\Struct)$ form the finite set of bounded obligations
--- a blocking (unknown) cell of an unsatisfied depth-$\le k$ clause instance of $\Thy$, or, once
$\Thy$ is locally satisfied, a goal cell to commit. Its size is bounded by the number of depth-$\le k$
ground instances, hence finite.
\item given a challenge with target (unknown) cell, the \emph{legal Builder replies}
$L_B(\Struct,\text{move})$ form the finite set of edits assigning that cell a value in its value set.
\item the \emph{Judge} is a total computable map $J(\Struct)\in\{\textsc{won},\textsc{theory\_failed},
\textsc{dead\_end},\textsc{open}\}$ obtained by running a bounded checker over the finite set of
depth-$\le k$ ground instances of $\Thy$ and the goal. In refute mode, $J(\Struct)=\textsc{won}$ iff
no checked instance of $\Thy$ is violated and the goal claim is falsified on the current commitments.
\end{itemize}
\end{definition}

\begin{definition}[Bounded episode]\label{def:episode}
A \emph{bounded episode} with node budget $b\in\N$ is the depth-first exploration, from the empty
structure, of the tree whose nodes are partial structures: at each $\textsc{open}$ node the Devil
selects one challenge from $L_D$, the Builder replies range over $L_B$, and each reply commits exactly
one previously-unknown cell (creating a child node). Each visited node is charged one unit against $b$.
The episode reports an outcome as in \Cref{def:outcomes}.
\end{definition}

\subsection{Termination}

\begin{proposition}[Finite episode termination]\label{prop:termination}
For fixed finite $\Sig$, finite domain $[n]$, depth $k$, the finite attack surface $L_D$ and finite
reply sets $L_B$ of \Cref{def:legal}, and a finite node budget $b$, every bounded episode halts and
returns exactly one of $\textsc{God\_Wins}$, $\textsc{Draw}$, or $\textsc{Devil\_Wins}$.
\end{proposition}

\begin{proof}
Each Builder reply commits one unknown cell and never un-commits one within an episode
(\Cref{def:episode}); since the total number of cells over $[n]$ is finite, every root-to-leaf path
has length at most that number, so the search tree has finite depth. At each node $L_D$ and $L_B$ are
finite, so the tree is finitely branching; by König's lemma (finite depth, finite branching) it is
finite. The exploration therefore visits finitely many nodes. It stops in exactly one of three ways:
(i) it reaches a node with $J=\textsc{won}$, reported $\textsc{God\_Wins}$; (ii) the node counter
reaches $b$ before any win, reported $\textsc{Draw}$; (iii) it exhausts the (finite) tree with no win
and without hitting $b$ --- the closed-frontier case --- whereupon the obstruction check of
\Cref{prop:devilsound} is run, reporting $\textsc{Devil\_Wins}$ if it certifies, else $\textsc{Draw}$.
These cases are mutually exclusive and exhaustive, and each terminates by finiteness.
\end{proof}

\subsection{Soundness of $\textsc{God\_Wins}$}

\begin{proposition}[$\textsc{God\_Wins}$ soundness, relative to the Judge]\label{prop:godsound}
If a bounded episode returns $\textsc{God\_Wins}$ with returned structure $\Struct^\star$, then
$\Struct^\star$ satisfies the checked fragment and the goal under the Judge specification: no depth-$\le k$
ground instance of $\Thy$ is violated on $\Struct^\star$, and (in refute mode) the goal claim is falsified
on $\Struct^\star$.
\end{proposition}

\begin{proof}
$\textsc{God\_Wins}$ is declared only at a node where $J(\Struct^\star)=\textsc{won}$
(\Cref{def:episode,def:legal}). By the definition of $J$, this holds iff the bounded checker finds no
violated depth-$\le k$ instance of $\Thy$ and finds the goal claim falsified on the current commitments.
The checker is a sound decision procedure over the \emph{explicitly given finite set} of depth-$\le k$
ground instances. Hence $\Struct^\star$ meets $\Thy$ and the goal up to that fragment. We make the
relativity explicit: this is soundness with respect to the Judge's checked fragment (depth-$\le k$
instances), not full first-order satisfaction.
\end{proof}

\subsection{Soundness of $\textsc{Devil\_Wins}$ (bounded obstruction)}

\begin{proposition}[$\textsc{Devil\_Wins}$ bounded-obstruction soundness]\label{prop:devilsound}
Suppose a bounded episode returns $\textsc{Devil\_Wins}$ with a bounded obstruction certificate
(\Cref{def:obstruction}) produced by exhaustive enumeration of $\mathrm{Comp}(\Struct)$ recording
$\textsc{winning\_completions}=0$ and $\textsc{budget\_exhausted}=\mathrm{false}$. Then no total
$\Sig$-structure on $[n]$ extending the current commitments is a Judge-accepted model achieving the goal;
a fortiori, no legal Builder continuation within the stated finite domain, attack surface, and budget
reaches a Judge-accepted state.
\end{proposition}

\begin{proof}
The certificate enumerates the finite set $\mathrm{Comp}(\Struct)$ (\Cref{def:completion}) and runs the
Judge on each completion, counting those that satisfy $\Thy$ and achieve the goal. $\textsc{winning\_%
completions}=0$ means no completion is a Judge-accepted model achieving the goal. By \Cref{def:episode}
every legal Builder continuation only commits unknown cells and therefore yields (or extends to) some
element of $\mathrm{Comp}(\Struct)$; and a state is Judge-accepted as a win ($J=\textsc{won}$) only if its
committed cells already determine such a winning completion. Since none exists, no continuation can be
accepted. The flag $\textsc{budget\_exhausted}=\mathrm{false}$ certifies that the enumeration was
completed --- the verdict is a genuine exhaustion of the finite completion space, not a budget timeout
(which would instead yield $\textsc{Draw}$). The certificate is independently re-checkable by re-running
the finite enumeration.
\end{proof}

\subsection{Exhaustive bounded completeness}

\begin{proposition}[Exhaustive bounded completeness]\label{prop:exhaustive}
For a finite domain $[n]$ and the finite completion space $\mathrm{Comp}(\Struct)$, exhaustive bounded
enumeration under the Judge returns either (a) a total completion that the Judge accepts as a model
achieving the goal (a $\textsc{God\_Wins}$ witness), or (b) a bounded obstruction certificate
(\Cref{def:obstruction}) establishing that none exists ($\textsc{Devil\_Wins}$). The procedure is
decidable and exhaustive over $\mathrm{Comp}(\Struct)$.
\end{proposition}

\begin{proof}
$\mathrm{Comp}(\Struct)$ is finite (\Cref{def:completion}). Enumerate it and run the Judge on each
element. If some completion is accepted, return it (case~a). Otherwise the enumeration log, with
zero winning completions and no budget exhaustion, is precisely the
bounded obstruction certificate (case~b), valid by \Cref{prop:devilsound}. Termination and
exhaustiveness are immediate from finiteness.
\end{proof}

\begin{remark}[This is not first-order completeness]\label{rem:not-fo-complete}
\Cref{prop:exhaustive} is completeness only relative to the \emph{fixed} finite domain $[n]$, the bounded
depth $k$, and the Judge's checked fragment. It says nothing about other domain sizes, larger depths, or
unbounded first-order satisfaction; it is decidability of a finite search, not a completeness theorem for
$\Thy$.
\end{remark}

\subsection{Neural admissibility (non-hallucination)}

\begin{proposition}[Neural admissibility / non-hallucination]\label{prop:neuraladmiss}
Suppose (a) the neural Devil selects its challenge only from $L_D(\Struct)$, (b) the neural God selects
its reply only from $L_B(\Struct,\text{move})$, and (c) every $\textsc{God\_Wins}$ / $\textsc{Devil\_Wins}$
verdict is emitted solely by the symbolic Judge and the obstruction enumeration of \Cref{prop:devilsound}.
Then replacing the symbolic selection heuristics by the neural policies does not affect the logical
soundness of $\textsc{God\_Wins}$ or $\textsc{Devil\_Wins}$: every such verdict remains valid in the sense
of \Cref{prop:godsound,prop:devilsound}. Only \emph{which} states are explored --- hence completeness
within a budget and efficiency --- can change.
\end{proposition}

\begin{proof}
By (a)--(b) the neural policies only re-order or choose \emph{among already-legal} moves; they never
introduce a move outside $L_D$ or a reply outside $L_B$, and by (c) they never emit a verdict. The verdicts
are produced by the Judge and the obstruction enumeration, whose soundness (\Cref{prop:godsound,prop:devilsound}) is stated for an arbitrary reached state / partial structure and does not reference how
the legal moves were chosen. Substituting symbolic selectors by neural selectors thus leaves every emitted
verdict valid. The explored subtree does depend on the policy, so completeness within a fixed budget and
running time may differ.
\end{proof}

\begin{remark}[Why this is the safety point]
\Cref{prop:neuraladmiss} is the mathematical reason the system can be ``neural'' without becoming
unsound: learning is confined to \emph{choosing among legal logical moves}, while truth, failure, and
obstruction are adjudicated only by the symbolic Judge. Logical soundness is therefore invariant under any
(adversarial or co-trained) replacement of the move-selection policies.
\end{remark}

\section{The ADAMANTIUM prototype}
\label{sec:apparatus}

The ADAMANTIUM prototype implements the God--Devil--Judge loop over finite partial
structures. We describe it factually as an experimental apparatus; the only
components that matter for the results are the mathematical/computational ones
listed below.\footnote{The public research snapshot corresponding to the prototype
reported here is the LOGAN repository branch
\url{https://github.com/Mircus/Logan/tree/adamantium-r7-r8-trainable-god}; the objects of this
paper are implemented under the namespace \texttt{logical\_gans.modelbuilder}.}

\subsection{Prototype 0 (symbolic-Devil scaffold)}
\label{sec:proto0}

An earlier scaffold, \emph{Prototype 0}, pairs a \emph{symbolic active Devil}
(selecting bounded challenges: blocking cells of unsatisfied depth-$\le k$ clause
instances, or goal cells to commit) with a \emph{learned Builder / God} (one
semantic edit per challenge, trained by imitation on winning trajectories) and a
\emph{symbolic Judge}. It operates over finite signatures, partial finite
$\Sig$-structures, and semantic edits ($\mathsf{SetRelation}$, $\mathsf{SetFunction}$,
$\mathsf{SetConstant}$), and emits $\textsc{God\_Wins}$ and $\textsc{Draw}$ only.

\begin{remark}[Prototype 0 is a partial instance]\label{rem:proto0-scaffold}
Prototype 0 realizes the construction/repair/certification loop of the UCMT game,
but its adversary is fixed (not learned) and it does not emit $\textsc{Devil\_Wins}$.
It is a faithful but partial instance of the game.  Prototype 1 adds the
learned Devil and the bounded obstruction certificate.
\end{remark}

\subsection{Prototype 1 (minimal judged adversarial demo)}
\label{sec:proto1}

\emph{Prototype 1} adds the missing organ --- a \emph{learned Devil} --- so that both
players are neural while the Judge stays symbolic. Its mathematical/computational
components are:
\begin{itemize}[leftmargin=1.5em]
\item legal Devil move enumeration ($L_D$ of \Cref{def:legal});
\item a $\mathsf{NeuralDevilPolicy}$ that scores the enumerated legal challenges;
\item legal Builder reply enumeration ($L_B$ of \Cref{def:legal});
\item a $\mathsf{NeuralGodPolicy}$ that scores the enumerated legal replies;
\item a symbolic Judge (the bounded checker of \Cref{def:legal});
\item the outcomes $\textsc{God\_Wins} / \textsc{Draw} / \textsc{Devil\_Wins}$ (\Cref{def:outcomes});
\item the bounded obstruction certificate of \Cref{def:obstruction} (exhaustive completion check).
\end{itemize}

\paragraph{One bounded round.}
\[
  \mathsf{GameState}
  \xrightarrow{\text{enumerate}} L_D
  \xrightarrow{\mathsf{NeuralDevilPolicy}} \text{chosen challenge}
  \xrightarrow{\text{enumerate}} L_B
\]
\[
  \xrightarrow{\mathsf{NeuralGodPolicy}} \text{chosen reply}
  \xrightarrow{\mathsf{SymbolicJudge}} \text{progress / witness / obstruction / next state}.
\]
By design the symbolic layer guarantees legality and verification while the neural
policies only \emph{score and select} among already-legal moves; the neural Devil
never invents an illegal attack. This discipline is exactly the hypothesis of the
non-hallucination guarantee (\Cref{prop:neuraladmiss}).

\section{Experiments}
\label{sec:experiments}

We report controlled cyclic experiments over
\[
  \Sig=\{E/2,\,s/1,\,a\},\qquad
  \Thy=\{\forall x\,E(x,s(x)),\ \forall x\,s(s(s(x)))=x\},
\]
with the goal ``refute $s(a)=a$''. These are illustrations of the bounded semantics on tiny
instances, not a benchmark study.

\begin{table}[t]
\centering
\small
\begin{tabular}{@{}p{0.16\linewidth}p{0.12\linewidth}p{0.27\linewidth}p{0.33\linewidth}@{}}
\toprule
Experiment & Domain & Policy mode & Observed outcome \\
\midrule
Demo A & $n=3$ & neural Devil / neural God & $\textsc{God\_Wins}$ with cyclic witness \\
Demo B & $n=2$ & neural Devil / neural God & $\textsc{Devil\_Wins}$; $128$ checked, $0$ wins \\
Co-training & $n=2,3$ & train both policies & both parameter sets updated \\
\bottomrule
\end{tabular}
\caption{Controlled cyclic experiments.  The table is not a benchmark claim: it records only the
small finite episodes used to validate the outcome semantics, obstruction certificate, and judged neural
co-training loop.}
\label{tab:cyclic-experiments}
\end{table}

\subsection{Demo~A: cyclic $\textsc{God\_Wins}$}

\paragraph{Setup.} Domain size $n=3$; goal: refute $s(a)=a$.

\paragraph{Mathematics.} A fixed-point-free $3$-cycle $s$ on $[3]$ (i.e.\ $s=(0\,1\,2)$)
satisfies $s^3=\mathrm{id}$ and has $s(a)\neq a$ for every $a$; taking $E=\{(x,s(x))\}$
realizes $\forall x\,E(x,s(x))$. Hence a $(k,b)$-model achieving the goal exists.

\paragraph{Observation.} The episode returns $\textsc{God\_Wins}$: the Judge accepts a
completion realizing the fixed-point-free $3$-cycle and refuting $s(a)=a$, consistent
with \Cref{prop:godsound}.

\subsection{Demo~B: cyclic $\textsc{Devil\_Wins}$ (certified)}

\paragraph{Setup.} Domain size $n=2$; goal: refute $s(a)=a$.

\paragraph{Mathematics.} On a $2$-element domain, $s^3=\mathrm{id}$ forces $s=\mathrm{id}$:
a transposition has order $2$, so $s^3=s\neq\mathrm{id}$. Therefore $s(a)=a$ is unavoidable
and no model refuting $s(a)=a$ exists on $[2]$.

\paragraph{Observation.} The episode returns $\textsc{Devil\_Wins}$ with a bounded
obstruction certificate obtained by exhaustively enumerating the finitely many completions
on the $2$-element domain. The certificate records:
\begin{center}
\begin{tabular}{ll}
\toprule
\textsc{completions\_checked} & $128$ \\
\textsc{winning\_completions} & $0$ \\
\textsc{budget\_exhausted}    & \texttt{false} \\
\textsc{judge\_verified}      & \texttt{true} \\
\bottomrule
\end{tabular}
\end{center}
This is the first prototype outcome in which the Devil's win is \emph{certified} (in the
sense of \Cref{prop:devilsound}), not merely a budget timeout: $\textsc{budget\_exhausted}$
is \texttt{false}, so the verdict is a genuine exhaustion of the completion space.

\subsection{Minimal judged co-training}

A minimal judged co-training loop updates \emph{both} a $\mathsf{NeuralDevilPolicy}$ and a
$\mathsf{NeuralGodPolicy}$ (a trainable policy/wrapper over legal moves), with rewards derived
solely from the symbolic Judge's outcome ($+1/-1$ to the winner/loser, a small negative to both
on $\textsc{Draw}$) and a one-step policy-gradient update of each policy. We verify that both
parameter sets actually change under training, that the Judge still decides every outcome, and
that the Demo~A ($\textsc{God\_Wins}$) and Demo~B ($\textsc{Devil\_Wins}$ with certificate) paths
are preserved.

\begin{remark}[Honest scope of the co-training result]
This is a toy controlled experiment, \emph{not} evidence of general finite model-finding
performance. On these tiny fixed tasks the Judge's outcome is determined by the task itself
(Demo~A is winnable, Demo~B is impossible), so neither learned policy can change the outcome;
the loop demonstrates only that both neural players participate, that rewards flow from the
symbolic Judge, and that both parameter sets are genuinely updated. By \Cref{prop:neuraladmiss}
this learning leaves the soundness of both verdicts intact.
\end{remark}

\section{Discussion: what UCMT is and is not}

\paragraph{Not classical satisfaction.}
UCMT does not claim $\Struct\models \Thy$; it claims ``$\Struct$ survives depth-$k$ scrutiny under budget $b$''.

\paragraph{A semantics with observables.}
Witnesses and repair costs are \emph{observables}. This is aligned with the ultrafinitist emphasis that
``truth'' should be something we can actually operationalize (bounded proof depth, bounded construction depth, etc.).

\section{From k-consistency to $(k,b)$-models (MTU-II bridge)}
\label{sec:kcons-bridge}

The finite UCMT results of \Cref{sec:metatheory} are self-contained. The connection to
MTU-II / Esenin--Volpin semantics developed in this section is \emph{conditional} and
remains a bridge principle rather than a theorem proved here: each result below is stated
relative to imported MTU-II constructions and explicit regime assumptions, which we do not
re-establish.

\subsection{Proof depth and bounded validity}

We write $T \vdash_k \varphi$ to mean: there exists a \emph{cut-free} proof tree of $\varphi$
from $T$ of depth $\le k$ (MTU-II's notion of proof complexity).
A theory $T$ is \emph{$k$-consistent} if there is no cut-free proof tree of a contradiction
from $T$ of depth $\le k$.

\begin{definition}[MTU-II bounded semantic validity]
Fix a depth parameter $k$. For an Esenin--Volpin model $M$ of depth at least $k$, write
$M \models_k \varphi$ (equivalently $M \Vdash_k \varphi$ in our notation) to mean:
$\varphi$ is forced/validated in $M$ \emph{up to depth $k$} (as in MTU-II).
\end{definition}

\subsection{MTU-II soundness/completeness at depth $k$}

The next two results are \emph{imported} from MTU-II \cite{Mannucci2023MTU2} and restated in our
notation; we do \emph{not} re-prove them here. Every claim in \Cref{sec:kcons-bridge} that depends
on them is therefore conditional on the MTU-II construction (the Saturation Lemma and Main Semantic
Lemma) holding for the relevant bounded fragment.

\begin{theorem}[Soundness at depth $k$ {\cite{Mannucci2023MTU2}}]
\label{thm:mtu-soundness}
If $\varphi$ has a cut-free proof of depth $< k$, then for every Esenin--Volpin model $M$
of depth at least $k$, we have $M \models_k \varphi$.
\end{theorem}

\begin{theorem}[Completeness at depth $k$ {\cite{Mannucci2023MTU2}}]
\label{thm:mtu-completeness}
If $\varphi$ is universally true in all Esenin--Volpin models up to depth $k$,
then $\vdash_k \varphi$.
\end{theorem}

\begin{remark}[Canonical model for $k$-consistent seeds]
MTU-II constructs a canonical Esenin--Volpin model $M_0$ by saturating depth-$m$ bi-theories
(Saturation Lemma + Main Semantic Lemma), yielding a universal counterexample provider
within the bounded fragment. In particular, the construction implies: given a $k$-consistent
seed, one can build a canonical model fragment that realizes it up to depth $k$.
\end{remark}

\subsection{A UCMT regime intended to match MTU-II}

Toward an existence result, we fix a regime \emph{intended} to align UCMT with MTU-II (whether the
alignment is exact is itself part of the research program, not asserted here): the \emph{Volpin regime}.

\begin{definition}[$\mathsf{Volpin}(k)$-Opponent]
\label{def:volpin-opponent}
Opponent challenges are of the form ``query $\varphi$ at a world of depth $\le k$,'' where
$\varphi$ ranges over closed formulas of structural complexity $\le k$ in the MTU-II bounded language
(the $m$-CF fragment with $m\le k$), and the admissible witness of failure is exactly a bounded
semantic counterexample in the sense of MTU-II (a world showing $\not\models_k$).
\end{definition}

\begin{definition}[$\mathsf{Volpin}(k)$-Builder]
\label{def:volpin-builder}
Builder maintains an Esenin--Volpin model of depth $\ge k$ (initially partial if desired),
and answers each challenge by either:
(i) producing the forcing derivation (a certificate that $M \models_k \varphi$), or
(ii) producing the bounded counterworld/witness that $M \not\models_k \varphi$.
\end{definition}

This is the UCMT game where the logic-observer is not an EF probe (LOGAN-style), but the
bounded forcing semantics itself. (EF opponents remain useful experimentally as \emph{heuristics}
for finding hard queries; see \Cref{sec:experiments}.)

\subsection{$k$-consistency $\Rightarrow$ existence of $(k,b)$-models}

\begin{proposition}[Conditional: $k$-consistency yields $(k,b)$-models in the Volpin regime]
\label{thm:kcons-exists-kb}
\emph{Assumptions:} (i) the MTU-II canonical construction (Saturation Lemma + Main Semantic Lemma)
is available for the bounded fragment; (ii) the $\mathsf{Volpin}(k)$ regime of
\Cref{def:volpin-opponent,def:volpin-builder}; (iii) a sound Judge; (iv) fixed depth $k$.
\emph{Then, conditionally:} if $T$ is $k$-consistent, there exists an Esenin--Volpin model $M$ of
depth $\ge k$ with $M \models_k T$, and consequently, for every budget $b\in\N$, Builder has a
winning strategy in $G(\Sig,T;k,b)$ under the $\mathsf{Volpin}(k)$ regime; equivalently, a
$(k,b)$-model of $T$ exists. We state this as a conditional proposition: it is only as strong as the
imported MTU-II construction and the regime assumptions, and is not independently proved here.
\end{proposition}

\begin{proof}[Proof sketch (conditional on the assumptions above)]
By MTU-II's canonical construction (Saturation Lemma + Main Semantic Lemma),
starting from a $k$-consistent seed for $T$ one can saturate to obtain a depth-$\le k$
replete bi-theory world and thus a canonical Esenin--Volpin model fragment in which the
intended $\Gamma$-component realizes $T$ up to depth $k$.
This yields an Esenin--Volpin model $M$ with $M \models_k T$.

Fix any budget $b$. In the $\mathsf{Volpin}(k)$ UCMT game, Builder plays ``stay inside $M$'':
each Opponent query is answered by the (already-defined) bounded forcing relation.
Since $M \models_k T$, Opponent cannot exhibit a bounded semantic failure for $T$ within the
allowed query language; hence Builder wins for all $b$.
\end{proof}

\subsection{Uniform failure $\Rightarrow$ bounded refutation}

\begin{proposition}[Conditional: uniform Opponent wins imply $k$-inconsistency]
\label{prop:uniform-failure-refutation}
\emph{Assumptions:} (i) the $\mathsf{Volpin}(k)$ regime; (ii) a complete bounded challenge language
(some $b_\star$ lets the Opponent enumerate the full depth-$\le k$ query space of the MTU-II fragment);
(iii) a sound Judge; (iv) MTU-II bounded completeness (\Cref{thm:mtu-completeness}) for that fragment.
\emph{Then, conditionally:} if for every Builder strategy the Opponent wins $G(\Sig,T;k,b_\star)$, then
$T$ is not $k$-consistent (there is a cut-free proof of contradiction from $T$ of depth $\le k$).
Without assumption (ii) a uniform Opponent win only yields a \emph{bounded obstruction certificate}
(\Cref{def:obstruction}), \emph{not} $k$-inconsistency; the lift to refutation is exactly what (ii)--(iv) buy.
\end{proposition}

\begin{proof}[Proof sketch (conditional on the assumptions above)]
If Opponent wins uniformly (with $b_\star$ covering the bounded query space), then \emph{no}
Esenin--Volpin model $M$ of depth $\ge k$ can satisfy $M \models_k T$; otherwise Builder could
play ``stay inside $M$'' and survive all depth-$k$ queries.

In MTU-II, the bounded soundness/completeness equivalence ties semantic universal validity
to $k$-derivability: the absence of any depth-$k$ semantic realization forces bounded derivability
of a contradiction (the bounded refutation tree is exactly the cut-free object controlled by $k$).
Formally, apply MTU-II completeness (contraposition) to the bounded fragment of the canonical model
construction: failure of semantic realizability at depth $k$ yields a cut-free refutation of depth $\le k$.
\end{proof}

\begin{remark}[What this would buy us, and what remains open]
\emph{If} the conditional propositions above hold (with their stated assumptions), UCMT would line up
with a bounded proof-theoretic threshold: $(k,b)$-models would exist essentially up to the
$k$-consistency boundary. We state this as a \emph{research-program goal}, not an established theorem:
the propositions are conditional on the imported MTU-II construction and on completeness of the bounded
challenge language, neither of which is re-established here. The honest current claim is the weaker,
implementable one --- bounded obstruction certificates on fixed finite domains
(\Cref{def:obstruction}, \Cref{prop:devilsound}) --- with the refutation lift left open.
\end{remark}

\section{Relation to LOGAN and finite model finding}

UCMT inherits its adversary from LOGAN (depth-bounded EF-style probes with
interpretable witnesses) and its bounded semantics from MTU-II (\Cref{sec:kcons-bridge}).
Relative to classical finite model finders --- Mace4 \cite{McCune2003Mace4},
Paradox \cite{ClaessenSorensson2003Paradox}, and SAT/SMT-backed search --- UCMT
differs in three ways: (i) modelhood is \emph{bounded and adversarial}
($\Struct \kbforces \Thy$), not classical satisfaction; (ii) the primitive output
is a small replayable \emph{witness} or \emph{bounded obstruction certificate}
(\Cref{def:obstruction}), not merely SAT/UNSAT; (iii) the Opponent can be a
\emph{learned} policy proposing hard challenges, with the symbolic Judge
guaranteeing soundness (\Cref{prop:neuraladmiss}). We do not claim to outperform
complete model finders; on tiny domains they are the natural oracle and a baseline
for the exhaustive obstruction checks.

\section{Limitations}
\label{sec:limitations}

We state the limitations of this work plainly.
\begin{itemize}[leftmargin=1.5em]
\item \textbf{The experiments are tiny.} Demos~A and B and the co-training loop run on $2$- and
$3$-element domains over one cyclic signature; they illustrate the semantics, not performance.
\item \textbf{The MTU-II bridge is conditional.} The results of \Cref{sec:kcons-bridge}
($k$-consistency $\Rightarrow (k,b)$-models, uniform failure $\Rightarrow$ $k$-inconsistency, exact
MTU-II matching) are stated relative to imported constructions and regime assumptions; they are not
proved here.
\item \textbf{Not a general theorem prover.} The Judge checks a bounded depth-$\le k$ fragment on a
fixed finite domain; soundness is relative to that fragment (\Cref{prop:godsound}).
\item \textbf{Not a complete finite model finder.} Exhaustive bounded completeness
(\Cref{prop:exhaustive}) is decidability of a finite search at fixed $(n,k)$, not first-order
completeness (\Cref{rem:not-fo-complete}).
\item \textbf{The co-training loop is minimal.} It is a one-step policy-gradient update of two
policies over legal moves on tasks whose outcomes are fixed; it is not adversarial training at scale.
\item \textbf{No commercial/product-level system is disclosed or claimed.} This paper concerns the
bounded semantics and a minimal prototype only.
\end{itemize}

\section{Conclusion}

UCMT replaces classical satisfaction by \emph{surviving bounded logical challenges} over finite stages, using three roles: Builder, Opponent, and symbolic Judge (also called God, Devil, and Judge in the implementation). Its
finite core is self-contained (\Cref{sec:metatheory}): every bounded episode terminates as
$\textsc{God\_Wins}$, $\textsc{Draw}$, or $\textsc{Devil\_Wins}$ (\Cref{prop:termination}); the first
two verdicts are sound relative to the Judge (\Cref{prop:godsound,prop:devilsound}); exhaustive bounded
enumeration is complete over finite completion spaces (\Cref{prop:exhaustive}); and confining learning to
the selection \emph{among legal moves} preserves soundness (\Cref{prop:neuraladmiss}). The prototype
realizes this loop and the controlled cyclic experiments of \Cref{sec:experiments} exhibit
$\textsc{God\_Wins}$ (Demo~A), a certified $\textsc{Devil\_Wins}$ bounded obstruction (Demo~B), and a
minimal judged co-training of both neural policies. The link to MTU-II / Esenin--Volpin proof theory
remains a conditional bridge (\Cref{sec:kcons-bridge}); establishing it for a complete bounded challenge
language, and scaling the experiments beyond the toy regime, are the natural next steps.

\nocite{*}
\bibliographystyle{alpha}
\bibliography{refs}

\end{document}